\documentclass[conference,final,twocolumn,a4paper]{IEEEtran}
\usepackage[latin1]{inputenc}
\usepackage{amsmath}
\usepackage{amsfonts}
\usepackage{amssymb}
\usepackage{graphicx}
\usepackage{tikz}
\usepackage{pgfplots}
\usepackage{color}
\usepackage{cancel}
\usepackage{nicefrac}
\usepackage[]{algorithm2e}
\usepackage{caption}
\usepackage{subcaption}
\usetikzlibrary{shapes,snakes}
\usetikzlibrary{calc}
\usetikzlibrary{patterns}
\usetikzlibrary{decorations.pathmorphing} 
\usetikzlibrary{arrows} 

\newtheorem{theorem}{Theorem} 
\newtheorem{lemma}[theorem]{Lemma}

\begin{document}
\newcommand*\circled[1]{\tikz[baseline=(char.base)]{
            \node[shape=circle,draw,inner sep=2pt] (char) {#1};}}
\setlength{\textfloatsep}{8pt plus 1.0pt minus 2.0pt}
\setlength{\intextsep}{1.5ex}

\newcommand{\basestation}[3]{
\coordinate (a) at (#1,#2);
\draw[line width=(#3)*1.5pt,scale=#3] ($(a)+(0, -1.08)$) -- ($(a)+(0, 1)$);
\draw[line width=(#3)*1.5pt,scale=#3] ($(a)+(0,1)$) .. controls ($(a)+(-0.25,-0.5)$) .. ($(a)+(-0.5,-0.9)$);
\draw[line width=(#3)*1.5pt,scale=#3] ($(a)+(0,1)$) .. controls ($(a)+(0.25,-0.5)$) .. ($(a)+(0.5,-0.9)$);
\draw[line width=(#3)*1.5pt,scale=#3] ($(a)+(0, -0.9)$) -- ($(a)+(-0.35, -0.7)$);
\draw[line width=(#3)*1.5pt,scale=#3] ($(a)+(0, -0.9)$) -- ($(a)+(0.35, -0.7)$);
\draw[line width=(#3)*0.75pt,scale=#3] ($(a)+(-0.35, -0.65)$) -- ($(a)+(0, -0.5)$);
\draw[line width=(#3)*0.75pt,scale=#3] ($(a)+(0.35, -0.65)$) -- ($(a)+(0, -0.5)$);
\draw[line width=(#3)*1pt,scale=#3] ($(a)+(0, -0.6)$) -- ($(a)+(-0.3, -0.45)$);
\draw[line width=(#3)*1pt,scale=#3] ($(a)+(0, -0.6)$) -- ($(a)+(0.3, -0.45)$);
\draw[line width=(#3)*0.5pt,scale=#3] ($(a)+(-0.3, -0.45)$) -- ($(a)+(0, -0.32)$);
\draw[line width=(#3)*0.5pt,scale=#3] ($(a)+(0.3, -0.45)$) -- ($(a)+(0, -0.32)$);
\draw[line width=(#3)*0.75pt,scale=#3] ($(a)+(0, -0.3)$) -- ($(a)+(-0.22, -0.17)$);
\draw[line width=(#3)*0.75pt,scale=#3] ($(a)+(0, -0.3)$) -- ($(a)+(0.22, -0.17)$);
\draw[line width=(#3)*0.5pt,scale=#3] ($(a)+(-0.22, -0.17)$) -- ($(a)+(0, -0.07)$);
\draw[line width=(#3)*0.5pt,scale=#3] ($(a)+(0.22, -0.17)$) -- ($(a)+(0, -0.07)$);;
\draw[line width=(#3)*0.75pt,scale=#3] (a) -- ($(a)+(-0.18, 0.11)$);
\draw[line width=(#3)*0.75pt,scale=#3] (a) -- ($(a)+(0.18, 0.11)$);
\draw[line width=(#3)*0.5pt,scale=#3] ($(a)+(-0.18, 0.11)$) -- ($(a)+(0,0.2)$);
\draw[line width=(#3)*0.5pt,scale=#3] ($(a)+(0.18, 0.11)$) -- ($(a)+(0,0.2)$);   
\draw[line width=(#3)*0.5pt,scale=#3] ($(a)+(0, 0.3)$) -- ($(a)+(-0.1, 0.37)$);
\draw[line width=(#3)*0.5pt,scale=#3] ($(a)+(0, 0.3)$) -- ($(a)+(0.1, 0.37)$);
\draw[line width=(#3)*0.25pt,scale=#3] ($(a)+(-0.1, 0.37)$) -- ($(a)+(0, 0.43)$);
\draw[line width=(#3)*0.25pt,scale=#3] ($(a)+(0.1, 0.37)$) -- ($(a)+(0, 0.43)$);
\draw[line width=(#3)*0.75pt,scale=#3] ($(a)+(0, 1.2)$) -- ($(a)+(0,1)$);
\draw[fill=white,scale=#3] ($(a)+(0, 1.2)$) circle (0.05cm);
\draw[line width=(#3)*1pt,decorate,decoration=expanding waves,decoration={segment length=(#3)*10pt},scale=#3] ($(a)+(0, 1.2)$) -- ($(a)+(0, 2.5)$);
}

\newcommand{\iPhone}[3]{
\coordinate (a) at (#1,#2);
\draw [line width=0.25pt,rounded corners=(#3)*1mm,fill=white,scale=(#3)] (a)--($(a)+(0.67,0)$)--($(a)+(0.67,1.381)$)--($(a)+(0,1.381)$)--cycle;
\draw [color=gray,line width=0.25pt,rounded corners=(#3)*0.8mm,fill=white,scale=(#3)] ($(a)+(0.015,0.015)$)--($(a)+(0.655,0.015)$)--($(a)+(0.655,1.366)$)--($(a)+(0.015,1.366)$)--cycle;
\draw [line width=0.25pt,rounded corners=(#3)*0.04mm,scale=(#3)] ($(a)+(0.2875,1.266)$)--($(a)+(0.3825,1.266)$)--($(a)+(0.3825,1.281)$)--($(a)+(0.2875,1.281)$)--cycle;
\draw[line width=0.25pt,scale=#3] ($(a)+(0.335,0.09)$) circle (0.055cm);
\draw[line width=0.25pt,scale=#3] ($(a)+(0.335,0.09)$) circle (0.044cm);
\draw[line width=0.25pt,scale=#3] ($(a)+(0.2275,1.2735)$) circle (0.015cm);
\draw[line width=0.25pt,scale=#3] ($(a)+(0.335,1.32)$) circle (0.01cm);
\draw [fill={rgb:black,1;white,4},line width=0.25pt,scale=(#3)] ($(a)+(0.042475,0.170195)$)--($(a)+(0.042475,0.170195)+(0.58505,0.0)$)--($(a)+(0.042475,0.170195)+(0.58505,1.04061)$)--($(a)+(0.042475,0.170195)+(0.0,1.04061)$)--cycle;
}

\title{On the Capacity of an Elemental Two-Way Two-Tier Network} 

\author{\begin{minipage}[c]{0.5\linewidth}
\centering
Dennis~Michaelis and Aydin~Sezgin\\
Institute of Digital Communication Systems\\
Ruhr-Universität Bochum\\
Email: \{dennis.michaelis, aydin.sezgin\}@rub.de
\end{minipage}
\begin{minipage}[c]{0.49\linewidth}
\centering
Eduard~A.~Jorswieck\\
Institute of Communications Theory\\
TU Dresden\\
Email: eduard.jorswieck@tu-dresden.de
\end{minipage}
}


\maketitle

\begin{abstract}
A basic setup of a two-tier network, where two mobile users exchange messages with a multi-antenna macrocell basestation, is studied from a rate perspective subject to beamforming and power constraints. The communication is facilitated by two femtocell basestations which act as relays as there is no direct link between the macrocell basestation and the mobile users. We propose a scheme based on physical-layer network coding and compute-and-forward combined with a novel approach that solves the problem of beamformer design and power allocation. We also show that the optimal beamformers are always a convex combination of the channels between the macro- and femtocell basestations. We then establish the cut-set bound of the setup to show that the presented scheme almost achieves the capacity of the setup numerically.\\
\IEEEpeerreviewmaketitle
\end{abstract}

\thanks{This work was supported by the DFG grants SE 1697/5 and JO 801/7-1.}

\section{Introduction}
\IEEEPARstart{N}{ext} generation (5G) wireless communication systems demand high-speed and high-quality data applications. Additionally, with the Internet of Things (IoT) as one of the new use cases for 5G, the number of devices exchanging information is expected to be 28 billion by 2020 \cite{Ericsson}. Introducing smaller sized cells, such as femto- and picocells, is considered as a potential solution for power efficiency while increasing the supported bitrates for devices in such cells \cite{PengCloudHetNet},\cite{EricssonHetNet}. Femto- and picocell basestations (BSs) can act as relays and forward messages from macrocell BSs to their destination, achieving peak data rates of up to $1$~Gbit/s in stationary or pedestrian environments under practical conditions \cite{PengHetNet}. However, although already deployed in 4G and LTE, the aforementioned advantages of relay networks are far from completely utilized \cite{PengCloudHetNet}.\\
An additional feature to enhance the performance even further is to operate the two-tier network in a two-way mode \cite{chaaban2011capacity},\cite{article113},\cite{Bla}, in which the exchange of information from the BS to the mobile users and vice versa is performed simultaneously.
\begin{figure}[!htb]
\begin{tikzpicture}
    \draw[fill=blue!20] (0,-0.5) ellipse (4cm and 1.5cm);
    \draw[fill=white] (-0.25,0.25) ellipse (2cm and 0.5cm);
    \draw[fill=white,rotate=5] (1.75,-1) ellipse (1.25cm and 0.5cm);
    \draw[fill=white,rotate=-10] (-1.5,-1.25) ellipse (1.75cm and 0.5cm);
    \iPhone{-1.25}{0.25}{0.3}
    \iPhone{1.5}{-1}{0.3} 
    \iPhone{0.8}{-1.1}{0.3} 
    \iPhone{1.2}{-1.25}{0.3} 
    \iPhone{-3}{-1}{0.3}
    \iPhone{-2.6}{-1.2}{0.3}
    \iPhone{-1.2}{-1.2}{0.3}
    \iPhone{-0.8}{-1.4}{0.3}
    \basestation{1}{0.5}{0.5}
    \basestation{2.25}{-0.6}{0.5}
    \basestation{-1.75}{-0.6}{0.5}
    \basestation{3}{0.85}{1}
    \draw[line width=1pt] (2.5,1.85) -- (3.5,1.85);
    \draw[line width=0.75pt] (3.5,2.05) -- (3.5,1.85);
    \draw[line width=0.75pt] (2.5,2.05) -- (2.5,1.85);
    \draw[fill=white] (2.5,2.05) circle (0.05cm);
    \draw[fill=white] (3.5,2.05) circle (0.05cm);
    \node at (3.25,2.05) {...};
    \node at (2.75,2.05) {...};
        
   \filldraw[fill=white, draw=black, line width=0.55pt] (-0.76,-0.06)  rectangle (-0.02,1.5);
   \node [draw,trapezium,trapezium left angle=30,trapezium right angle=-30, minimum width=1.2cm,fill=white] at (-0.15,1.6275) {};
   \node [draw,trapezium,trapezium left angle=-60,trapezium right angle=60, minimum height=0.3cm,fill=white, rotate=90] at (0.225,0.85) {\phantom{blabla.l}};    
   \filldraw[fill=white,draw=black,line width=0.55pt] (-0.7,1.1) rectangle(-0.5,1.3);    
   \filldraw[fill=white,draw=black,line width=0.55pt] (-0.5,1.1) rectangle(-0.3,1.3);
   \filldraw[fill=white,draw=black,line width=0.55pt] (-0.3,1.1) rectangle(-0.1,1.3);

   \filldraw[fill=white,draw=black,line width=0.55pt] (-0.7,0.8) rectangle(-0.5,1);    
   \filldraw[fill=white,draw=black,line width=0.55pt] (-0.5,0.8) rectangle(-0.3,1);
   \filldraw[fill=white,draw=black,line width=0.55pt] (-0.3,0.8) rectangle(-0.1,1);

   \filldraw[fill=white,draw=black,line width=0.55pt] (-0.7,0.5) rectangle(-0.5,0.7);    
   \filldraw[fill=white,draw=black,line width=0.55pt] (-0.5,0.5) rectangle(-0.3,0.7);
   \filldraw[fill=white,draw=black,line width=0.55pt] (-0.3,0.5) rectangle(-0.1,0.7);
   	
   \filldraw[fill=white, draw=black, line width=0.55pt] (-0.65,-0.08) rectangle (-0.45,0.3);    
\end{tikzpicture}
\caption{Example of a two-tier network.}
\label{fig:Intro}
\end{figure}
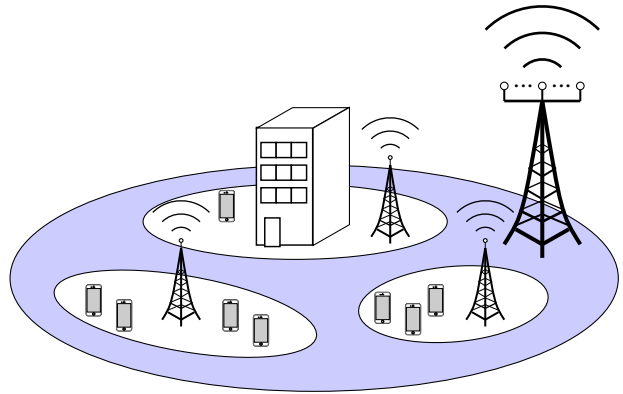

An actual example of a two-way two-tier network, consisting of one macrocell BS, multiple femtocell BSs and mobile users, is depicted in Fig. \ref{fig:Intro}. In order to achieve progress on the understanding of the fundamental limits of this rather involved setup, we focus our study on an elemental subsystem of the overall setup. To this end, we consider a wireless communication system with one macrocell BS and two mobile users, each of which is served by a femtocell BS. The macrocell BS has to deliver independent information towards each of the mobile users, respectively. To do so, the macrocell BS is equipped with an array of multiple antennas. Existing techniques for canceling non-causally known interference can be utilized to eliminate interference without power penalty \cite{Costa}, \cite{DirtyLattices}. The mobile users on the other hand are equipped with only a single antenna, and are also to deliver information to the macrocell BS. Because there is no direct path between the macrocell BS and the mobile users, these information have to be received and forwarded by the two femtocell BS, which use nested lattice codes to effectively process multiple messages at a time \cite{LatticesRelay}, \cite{Anas}. Upon reception, both the mobile users and the macrocell basestation can use the knowledge of their transmitted messages as side information to improve the performance.

\section{System Model}\label{section2}

This paper considers the bidirectional two-way communication in a two-tier network as shown in Fig. \ref{fig:Communication System}, which is an elementary component of the heterogeneous network system depicted in Fig. \ref{fig:Intro}. The system consists of one macrocell basestation ($BS$), two mobile users ($U_1$ and $U_2$) and two femtocell BSs ($R_1$ and $R_2$), which act as relays to aid the communication process. It is assumed that the macrocell BS is equipped with $M\geq 2$ antennas, while each femtocell BS and each mobile user is only equipped with a single antenna. All devices operate in half-duplex mode. It is further assumed that each femtocell BS only serves a single mobile user and all participants have perfect channel state information (CSI). The details of the studied network are given as follows.\\
There are four messages $m_1,\dots,m_4$ in total to be transmitted and all transmissions occur over quasi-static slow-fading channels $\boldsymbol{h}_1,\boldsymbol{h}_2,h_3,h_4$. The channels between the macrocell BS and the femtocell BSs are described by $\boldsymbol{h}_1,\boldsymbol{h}_2\in\mathbb{C}^{M\times1}$, whereas $h_3,h_4\in\mathbb{C}$ denote the channels between the femtocell BSs and the users, respectively. Messages $m_1$ and $m_2$ are originated at the macrocell BS and are destined to be transmitted to $U_1$ and $U_2$, respectively, whereas messages $m_3$ and $m_4$ are received by the macrocell BS, coming from the mobile users (see Fig. \ref{fig:Communication System}).\\
We may divide the whole transmission process into two phases. In the \textit{uplink phase} all messages $m_1,\dots,m_4$ are sent from their origin (macrocell BS, mobile user) to the respecting femtocell BS. After some processing, the femtocell BSs will forward the messages to their destinations (macrocell BS, mobile user). We consequently refer to this phase as the \textit{downlink phase}.
 
\begin{figure}[!htb]
\centering
\begin{tikzpicture}
\basestation{0}{-0.2}{0.8}
\draw[line width=0.5pt] (-0.4,0.6) -- (0.4,0.6);
\draw[line width=0.5pt] (0.4,0.6) -- (0.4,0.7);
\draw[line width=0.5pt] (-0.4,0.6) -- (-0.4,0.7);
\draw[fill=white] (-0.4,0.75) circle (0.04cm);
\draw[fill=white] (0.4,0.75) circle (0.04cm);
\node at (0.2,0.7) {...};
\node at (-0.2,0.7) {...};
\basestation{-2}{-0.5}{0.5}
\basestation{2}{-0.5}{0.5}
\iPhone{-3.5}{-1}{0.5}
\iPhone{3.2}{-1}{0.5}
\draw[line width=0.75pt, <->,>=latex](-0.5, -0.5)-- node[above=0.2mm]{$\boldsymbol{h}_1$}(-1.75, -0.5);
\draw[line width=0.75pt, <->,>=latex](0.5, -0.5)--node[above=0.2mm]{$\boldsymbol{h}_2$}(1.75, -0.5);
\draw[line width=0.75pt, <->,>=latex]
(-3.1, -0.5) --  node[above=0.2mm]{$h_3$} (-2.2, -0.5);
\draw[line width=0.75pt, <->,>=latex]
(3.1, -0.5) --  node[above=0.2mm]{$h_4$} (2.2, -0.5);
\draw[->,>=latex,densely dotted,line width=0.75pt] (-0.75,1.2) .. controls (-1.1,1.9) and (-1.7,1.2) .. node[above=0.2mm]{$m_1$}(-1.825,0.8);
\draw[->,>=latex,dashed] (-2.25,0.85) .. controls (-2.375,1.2) and (-3.125,1.5) .. node[above=0.2mm]{$m_1$}(-3.25,0);
\draw[->,>=latex,densely dotted,line width=0.75pt] (0.75,1.2) .. controls (1.1,1.9) and (1.7,1.2) .. node[above=0.2mm]{$m_2$}(1.825,0.8);
\draw[->,>=latex,dashed] (2.25,0.85) .. controls (2.375,1.2) and (3.125,1.5) .. node[above=0.2mm]{$m_2$}(3.25,0);
\draw[<-,>=latex,densely dotted,line width=0.75pt] (-2.125,-1.2) .. controls (-2.25,-1.6) and (-3.25,-1.6) .. node[below=0.2mm]{$m_3$}(-3.25,-1.2);
\draw[<-,>=latex,dashed] (-0.125,-1.2) .. controls (-0.25,-1.6) and (-1.7,-1.6) .. node[below=0.2mm]{$m_3$}(-1.75,-1.2);
\draw[<-,>=latex,densely dotted,line width=0.75pt] (2.125,-1.2) .. controls (2.25,-1.6) and (3.25,-1.6) .. node[below=0.2mm]{$m_4$}(3.25,-1.2);
\draw[<-,>=latex,dashed] (0.125,-1.2) .. controls (0.25,-1.6) and (1.7,-1.6) .. node[below=0.2mm]{$m_4$}(1.75,-1.2);
\node at (0,-1.8) {$BS$};
\node at (-1.95,-1.8) {$R_1$};
\node at (1.95,-1.8) {$R_2$};
\node at (-3.25,-1.8) {$U_1$};
\node at (3.45,-1.8) {$U_2$};
\draw [fill=white] (-3.75,2) rectangle (3.75,2.5);
\draw[line width=0.75pt,densely dotted] (-3.5,2.25) -- (-2.75,2.25);
\node at (-1.5,2.22) {Uplink Phase};
\draw[densely dashed] (0,2.25) -- (0.75,2.25);
\node at (2.25,2.25) {Downlink Phase};
\end{tikzpicture}
\caption{System model with slow fading channels $\boldsymbol{h}_1,\dots,h_4$ and flow of messages $m_1,\dots,m_4$.}
\label{fig:Communication System}
\end{figure}
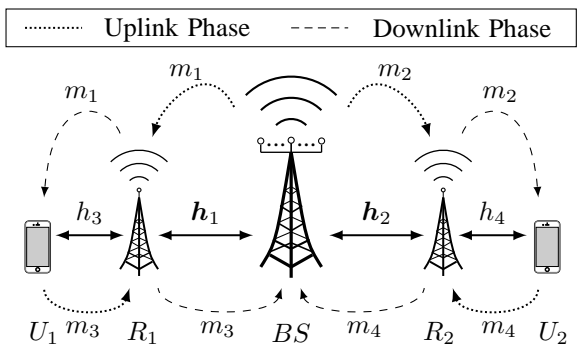

\subsection{The Uplink Phase}

This is the part of the transmission process in which all messages are to be sent towards the femtocell BSs. Since the femtocell BSs' task is to aid the communication process, they act as relays ($R_1$ and $R_2$). The baseband complex symbols received at $R_1$ and $R_2$ can therefore be expressed, respectively, as
\begin{align}
y_{R_1} = \boldsymbol{h}_1^H\boldsymbol{x}_{BS}+h_3\sqrt{P_{U_1}}u_{m_3}+z_{R_1}\label{yr1},\\
y_{R_2} = \boldsymbol{h}_2^H\boldsymbol{x}_{BS}+h_4\sqrt{P_{U_2}}u_{m_4}+z_{R_2}\label{yr2},
\end{align}
by making use of the possibility to transmit the lattice codewords $u_{m_1},u_{m_2}$ via beamforming
\begin{align}
\boldsymbol{x}_{BS}=\boldsymbol{w}_1u_{m_1}+\boldsymbol{w}_2u_{m_2},
\end{align}
using a random encoding \cite{Cover} of the messages $m_i$ onto $u_{m_i}$. Note that the encoding does not take into account the received signals from previous channel uses, thus restrictive encoding \cite{Anas} is applied. Naturally, the performance could be improved by using non-restrictive encoding, i.e., utilizing the received signals for the encoding. This, however, is not the focus of this work. The noises at the femtocell BSs are denoted by $z_{R_1},z_{R_2}$ and without loss of generality all noises are assumed to be additive white Gaussian noises (AWGN) with zero-mean and unit-variance. The transmit beamformers are denoted by $\boldsymbol{w}_1,\boldsymbol{w}_2\in\mathbb{C}^{M\times 1}$ and are subject to optimization. Also note that the utilized power of the macrocell BS is $P_{BS}=||\boldsymbol{w}_1||^2+||\boldsymbol{w}_2||^2$, and the power consumed by the mobile users are $P_{U_1}$ and $P_{U_2}$, respectively. \\
Upon reception, the femtocell BSs decode and process the signals by applying compute-and-forward \cite{LatticesRelay} to generate the transmit symbols for the downlink phase in the next phase.

\subsection{The Downlink Phase}
Once the transmit symbols are computed, the femtocell BSs will forward these signals to the macrocell BS and the mobile users. We refer to this step as the downlink phase. Because the femtocell BSs deploy only a single antenna, this step is done by a simple broadcast, leading to
\begin{align}
y_{U_1} = h_3\sqrt{P_{R_1}}u_{R_1}+z_{U_1}\label{yu1},\\
y_{U_2} = h_4\sqrt{P_{R_2}}u_{R_2}+z_{U_2}\label{yu2},
\end{align}
where $z_{U_1},z_{U_2}$ is the noise at the mobile users and $u_{R_1},u_{R_2}$ are the transmit symbols picked from a Gaussian codebook.\\
In contrast to the mobile users, the macrocell BS receives the signals from both femtocell BSs at the same time. We therefore can express the received signal at the macrocell BS as
\begin{align*}
\boldsymbol{y}_{BS} = \boldsymbol{h}_1\sqrt{P_{R_1}}u_{R_1}+\boldsymbol{h}_2\sqrt{P_{R_2}}u_{R_2}+\boldsymbol{z}_{BS},\label{ybs}
\end{align*}
where $\boldsymbol{y}_{BS},\boldsymbol{z}_{BS}\in\mathbb{C}^{M\times 1}$ and $\boldsymbol{z}_{BS}\sim\mathcal{CN}(\boldsymbol{0},\boldsymbol{I})$ is the AWGN at the macrocell BS.\\
To obtain an estimate for the transmitted symbols we write
\begin{equation*}
\begin{aligned}
\hat{u}_{m_3}=&\boldsymbol{v}_1^H\boldsymbol{y}_{BS},\\
\hat{u}_{m_4}=&\boldsymbol{v}_2^H\boldsymbol{y}_{BS},
\end{aligned}
\end{equation*}
where $\boldsymbol{v}_1,\boldsymbol{v}_2\in\mathbb{C}^{M\times 1}$ are receive beamformers. Without loss of generality it is assumed that $||\boldsymbol{v}_1||,||\boldsymbol{v}_2||=1$.\\
In the following we first state the main results before we derive the data rate expressions step by step to formulate the overall optimization problem.

\section{Main Results}\label{Section3}

The main contribution of this work is summarized below.\\
In the following we derive rate expressions for both the proposed compute-and-forward strategy and the cut-set bounds of the setup depicted in Fig. \ref{fig:Communication System}. We numerically compare our achievability scheme to the upper bounds to illustrate that the proposed strategy is close to optimal.\\
We also show that the optimal beamforming strategy consist of beamformers which lie in the span of $\boldsymbol{h}_1$ and $\boldsymbol{h}_2$. With this result we are able to construct beamformers appropriately. 

\section{Data Rates and Problem Formulation}\label{Section4}

\subsection{Uplink Phase}

The transmission of messages $m_1$ and $m_2$ from the macrocell BS to the femtocell BSs represents a broadcast channel. Thus, by using \textit{Dirty Paper Coding} \cite{Costa}, \cite{DirtyLattices}, we can cancel non-causally known interference without power penalty. Thus, exploiting the properties of codewords from nested lattices \cite{LatticesRelay}, the following signal-to-interference-plus-noise ratios (SINRs) for messages $m_1,m_2$ can be achieved
\begin{align}
\gamma_{1,UL}&=\frac{|\boldsymbol{w}_1^H\boldsymbol{h}_1|^2}{1+b_{EO}\cdot|\boldsymbol{w}_2^H\boldsymbol{h}_1|^2},\label{EO1}\\
\gamma_{2,UL}&=\frac{|\boldsymbol{w}_2^H\boldsymbol{h}_2|^2}{1+(1-b_{EO})\cdot|\boldsymbol{w}_1^H\boldsymbol{h}_2|^2}\label{EO2},
\end{align}
where $b_{EO}\in\{0,1\}$ is a binary variable indicating the encoding order (EO) at the macrocell BS. Please note that the second index denotes the uplink phase (UL) and that $b_{EO}=1$ refers to the case where message $m_1$ is encoded first.\\
When the mobile users $U_1$ and $U_2$ transmit towards the femtocell BSs in the uplink phase, messages $m_3$ and $m_4$ are observed at the femtocell BS with a SINR of
\begin{align}
\gamma_{3,UL}&=\frac{|h_3|^2P_{U_1}}{1+b_{EO}\cdot|\boldsymbol{w}_2^H\boldsymbol{h}_1|^2}\label{ULSINR1},\\
\gamma_{4,UL}&=\frac{|h_4|^2P_{U_2}}{1+(1-b_{EO})\cdot|\boldsymbol{w}_1^H\boldsymbol{h}_2|^2}\label{ULSINR2},
\end{align}
where $P_{U_1},P_{U_2}$ are the transmit powers of the mobile users, respectively.\\
By transferring the results of \cite{LatticesRelay}, \cite{TwoWay}, \cite{Anas} to the considered setup, we can express the uplink data rates for messages $m_1,\dots,m_4$ as follows
\begin{align}
R_{m_1,UL}&\leq\frac{1}{2}\biggl\lbrack\log_2\Bigl(\frac{\gamma_{1,UL}}{\gamma_{1,UL}+\gamma_{3,UL}}+\gamma_{1,UL}\Bigl)\biggr\rbrack^+,\label{UL1}\\
R_{m_2,UL}&\leq\frac{1}{2}\biggl\lbrack\log_2\Bigl(\frac{\gamma_{2,UL}}{\gamma_{2,UL}+\gamma_{4,UL}}+\gamma_{2,UL}\Bigl)\biggr\rbrack^+,\\
R_{m_3,UL}&\leq\frac{1}{2}\biggl\lbrack\log_2\Bigl(\frac{\gamma_{3,UL}}{\gamma_{1,UL}+\gamma_{3,UL}}+\gamma_{3,UL}\Bigl)\biggr\rbrack^+,\\
R_{m_4,UL}&\leq\frac{1}{2}\biggl\lbrack\log_2\Bigl(\frac{\gamma_{4,UL}}{\gamma_{2,UL}+\gamma_{4,UL}}+\gamma_{4,UL}\Bigl)\biggr\rbrack^+,\label{UL4}
\end{align}
with $\gamma_{1,UL},\dots,\gamma_{4,UL}$ from (\ref{EO1}) - (\ref{ULSINR2}) respectively, and $[\cdot]^+\overset{\Delta}{=}\max\{\cdot,0\}$. Note that the factor of $\frac{1}{2}$ comes from the fact that all devices operate in half duplex mode \cite{Factor}.

\subsection{Downlink Phase}

The femtocell BSs broadcast their codewords which are based on the compute-and-forward \cite{CAFPro} processing of their received signals in the uplink phase. The mobile users and the macrocell BS can use their messages as side information to improve the performance (see \cite{MyPaper} for details). We obtain the following SNRs at the mobile users
\begin{align}
\gamma_{1,DL}&=|h_3|^2P_{R_1},\label{DO1}\\
\gamma_{2,DL}&=|h_4|^2P_{R_2},
\end{align}
where $P_{R_1},P_{R_2}$ are the transmit powers of the femtocell BSs, respectively.\\
At the macrocell BSs successive decoding is applied \cite{BCMACDual}. This allows to cope with some of the interference without further power penalty, leading to
\begin{align}
\gamma_{3,DL}&=\frac{|\boldsymbol{v}_1^H\boldsymbol{h}_1|^2P_{R_1}}{1+b_{DO}\cdot|\boldsymbol{v}_1^H\boldsymbol{h}_2|^2P_{R_2}},\\
\gamma_{4,DL}&=\frac{|\boldsymbol{v}_2^H\boldsymbol{h}_2|^2P_{R_2}}{1+(1-b_{DO})\cdot|\boldsymbol{v}_2^H\boldsymbol{h}_1|^2P_{R_1}},\label{DO2}
\end{align}
where $b_{DO}\in\{0,1\}$ is a binary variable indicating the decoding order (DO). Note that $b_{DO}=1$ refers to the case where codeword $u_{m_3}$ is decoded first.\\
To summarize the downlink phase up, we can express the rates for messages $m_1,\dots,m_4$ as
\begin{align}
R_{m_i,DL}\leq\frac{1}{2}\log_2(1+\gamma_{i,DL}),~i=1,\dots,4.
\end{align}

\subsection{Overall Rate Expressions and Problem Formulation}

The overall rate of an arbitrary message in Fig. \ref{fig:Communication System} with the presented approach is upper-bounded by the minimum of its uplink and downlink rate
\begin{align}
R_{m_i}&\leq\min\bigl\{R_{m_i,UL},R_{m_i,DL}\bigr\},~i=1,\dots,4,\label{Theorem1}
\end{align}
which can be seen from the fact that always the weakest link determines the maximum rate from and towards a user \cite{NIT}. We then formulate the sum-rate max-min problem of optimal power allocation and beamformer design as follows
\begin{equation}
\begin{array}{cl}
\underset{R_{m_i},P_j,\boldsymbol{w}_k,\boldsymbol{v}_l,b_n}{\text{maximize}}\label{OptPro}
& \underset{i=1}{\overset{4}{\sum}}\min\bigl\{R_{m_i,UL},R_{m_i,DL}\bigr\} \\
\text{subject to}
& (\ref{Theorem1}),~0\leq P_{j}\leq P_{j,\text{MAX}},\\
& \underset{k}{\sum}||\boldsymbol{w}_k||^2\leq P_{BS,{\text{MAX}}},\\
& ||\boldsymbol{v}_l||^2=1,~b_{n}\in\{0,1\},
\end{array}
\end{equation}
where $j=U_1,\dots,R_2,~k=1,2,~l=1,2$ and $n=EO,DO$. $P_{j,\text{MAX}},P_{BS,\text{MAX}}\in\mathbb{R}^+$ are the maximum available powers at the mobile users, the femtocell BSs and the macrocell BS, respectively.

\section{Zero-Forcing and Efficient Power Allocation}\label{section5}
To translate the max-min problem into a simpler maximization-only problem, we rewrite constraints (\ref{Theorem1}) with the help of
\begin{equation}\label{Step1}
\begin{aligned}
   R_{m_i}&\leq\min\{R_{m_i,UL},R_{m_i,DL}\}\\
   \Leftrightarrow R_{m_i}&=
      \begin{cases}
     R_{m_i}\leq R_{m_i,UL}\\
     R_{m_i}\leq R_{m_i,DL}
   \end{cases}.
\end{aligned}
\end{equation}
Secondly, we simplify the fraction terms in (\ref{UL1}) - (\ref{UL4}) that were introduced due to the nested lattice codes. By choosing identical lattices we get
\begin{align}
R_{m_i,UL}\leq\frac{1}{2}\biggl\lbrack\log\Bigl(\frac{1}{2}+\gamma_{i,UL}\Bigl)\biggr\rbrack^+,~i=1,\dots,4,\label{Step2}
\end{align}
which gives a lower bound on the proposed achievability scheme.\\
Before we proceed, we make the following statement.\newpage
\begin{lemma}\label{FirstTheorem}
The optimal beamformers fulfill $\boldsymbol{w}_1^*,\boldsymbol{w}_2^*\in\operatorname{span}(\boldsymbol{h}_1,\boldsymbol{h}_2)$, which is only a two-dimensional subspace of $\mathbb{C}^M$.
\end{lemma}
\begin{proof}[Proof]
See \cite{MyPaper}.
\end{proof}
By the help of Lemma \ref{FirstTheorem} we are able to construct the zero-forcing transmit beamformers as below.
\begin{align}
\boldsymbol{w}_{1,ZF} = 
\begin{cases} \sqrt{P_1}\cdot\frac{\boldsymbol{h_1}}{||\boldsymbol{h}_1||} &\mbox{if } b_{EO}=0, \label{Step2.5}\\
\sqrt{P_1}\cdot\frac{\boldsymbol{h_{1,\boldsymbol{h}_2^\bot}}}{||\boldsymbol{h_{1,\boldsymbol{h}_2^\bot}}||} &\mbox{if } b_{EO}=1,
\end{cases}\\
\boldsymbol{w}_{2,ZF} = 
\begin{cases} 
\sqrt{P_2}\cdot\frac{\boldsymbol{h_{2,\boldsymbol{h}_1^\bot}}}{||\boldsymbol{h_{2,\boldsymbol{h}_1^\bot}}||} &\mbox{if } b_{EO}=0,\\
\sqrt{P_2}\cdot\frac{\boldsymbol{h_2}}{||\boldsymbol{h}_2||} &\mbox{if } b_{EO}=1,
\end{cases}
\end{align}
with $P_1+P_2\leq P_{BS,\text{MAX}}$ being optimization parameters, indicating the power distribution of the macrocell BS among its beamformers, and
\begin{equation*}
\begin{aligned}
\boldsymbol{h_{1,\boldsymbol{h}_2^\bot}}&=\boldsymbol{h}_1-\frac{\boldsymbol{h}_2^H\boldsymbol{h}_1}{||\boldsymbol{h}_2||^2}\boldsymbol{h}_2,\\
\boldsymbol{h_{2,\boldsymbol{h}_1^\bot}}&=\boldsymbol{h}_2-\frac{\boldsymbol{h}_1^H\boldsymbol{h}_2}{||\boldsymbol{h}_1||^2}\boldsymbol{h}_1.
\end{aligned}
\end{equation*}
We construct the zero-forcing receive beamformers in a similar fashion by
\begin{align}
\boldsymbol{v}_{1,ZF} = 
\begin{cases} \frac{\boldsymbol{h_1}}{||\boldsymbol{h}_1||} &\mbox{if } b_{DO}=0, \\
\frac{\boldsymbol{h_{1,\boldsymbol{h}_2^\bot}}}{||\boldsymbol{h_{1,\boldsymbol{h}_2^\bot}}||} &\mbox{if } b_{DO}=1,
\end{cases}\\
\boldsymbol{v}_{2,ZF} = 
\begin{cases} 
\frac{\boldsymbol{h_{2,\boldsymbol{h}_1^\bot}}}{||\boldsymbol{h_{2,\boldsymbol{h}_1^\bot}}||} &\mbox{if } b_{DO}=0,\\
\frac{\boldsymbol{h_2}}{||\boldsymbol{h}_2||} &\mbox{if } b_{DO}=1.\label{Step3}
\end{cases}
\end{align}
Note that without loss of generality we fixed $||\boldsymbol{v}_1||,||\boldsymbol{v}_2||$ to be equal to 1. Please also note that by designing the beamformers according to (\ref{Step2.5})-(\ref{Step3}), we are able to cancel all remaining interference that is present during the transmission process.\\
Combining (\ref{Step1}) - (\ref{Step3}), we can rewrite constraints (\ref{Theorem1}) as
\begin{equation}\label{EasyConstraint}
\begin{aligned}
R_{m_1}&\leq0.5\Bigl[\log_2(\nicefrac{1}{2}+|\boldsymbol{h}_1^H\boldsymbol{w}_{1,ZF}|^2)\Bigr]^+\\
R_{m_1}&\leq0.5\log_2(1+|h_3|^2P_{R_1})\\
R_{m_2}&\leq0.5\Bigl[\log_2(\nicefrac{1}{2}+|\boldsymbol{h}_2^H\boldsymbol{w}_{2,ZF}|^2)\Bigr]^+\\
R_{m_2}&\leq0.5\log_2(1+|h_4|^2P_{R_2})\\
R_{m_3}&\leq0.5\Bigl[\log_2(\nicefrac{1}{2}+|h_3|^2P_{U_1})\Bigr]^+\\
R_{m_3}&\leq0.5\log_2(1+|\boldsymbol{h}_1^H\boldsymbol{v}_{1,ZF}|^2P_{R_1})\\
R_{m_4}&\leq0.5\Bigl[\log_2(\nicefrac{1}{2}+|h_4|^2P_{U_2})\Bigr]^+\\
R_{m_4}&\leq0.5\log_2(1+|\boldsymbol{h}_2^H\boldsymbol{v}_{2,ZF}|^2P_{R_2}).
\end{aligned}
\end{equation}
By the nature of these constraints we see that setting 
\begin{equation}\label{Power}
P_j=P_{j,\text{MAX}}
\end{equation}
will result in the highest overall rate. Note that there might be many different power allocations to achieve the maximum rate of the setup, which can be seen from (\ref{Theorem1}). Our goal, however, is to find the most efficient allocation to achieve this rate.\\
To this end, we firstly determine the maximum throughput of the system, by rewriting (\ref{OptPro}) as

\begin{equation}
\begin{array}{cl}
\underset{R_{m_i},P_1,P_2,b_n}{\text{maximize}}\label{FinalOpt2}
& \underset{i=1}{\overset{4}{\sum}}R_{m_i} \\
\text{subject to}
& (\ref{Step2.5})-(\ref{Step3}),(\ref{EasyConstraint}),(\ref{Power}),\\
& P_1+P_2=P_{BS,\text{MAX}},\\
& b_{n}\in\{0,1\}, ~n=EO,DO,
\end{array}
\end{equation}
which is a convex optimization problem once the en- and decoding order are fixed.\\
Secondly, after (\ref{FinalOpt2}) is solved, we set 
\begin{equation}
\begin{aligned}
P_{BS}^{\text{opt}}&=\max\Bigl\{\frac{2^{R_1}}{|\boldsymbol{h}_1^H\boldsymbol{w}_{1,ZF}|^2},\frac{2^{R_2}}{|\boldsymbol{h}_2^H\boldsymbol{w}_{2,ZF}|^2}\Bigr\},\\
P_{R_1}^{\text{opt}}&=\max\Bigl\{\frac{2^{R_1}-\nicefrac{1}{2}}{|h_3|^2},\frac{2^{R_3}-\nicefrac{1}{2}}{|\boldsymbol{h}_1^H\boldsymbol{v}_{1,ZF}|^2}\Bigr\},\\
P_{R_2}^{\text{opt}}&=\max\Bigl\{\frac{2^{R_2}-\nicefrac{1}{2}}{|h_4|^2},\frac{2^{R_4}-\nicefrac{1}{2}}{|\boldsymbol{h}_2^H\boldsymbol{v}_{2,ZF}|^2}\Bigr\},\\
P_{U_1}^{\text{opt}}&=\frac{2^{R_3}}{|h_3|^2},~P_{U_1}^{\text{opt}}=\frac{2^{R_4}}{|h_4|^2},
\end{aligned}
\end{equation}
which will give the smallest required powers to achieve the maximum sum rate and can be seen from (\ref{EasyConstraint}). We will refer to this approach as the zero-forcing efficient power allocation (ZF-EPA) scheme from now on.\\
It is yet to investigate how well ZF-EPA performs with regard to the cut-set upper bound, which will be established in the following section.

\section{Cut-set Bound}\label{Section6}

The cut-set bound will give an upper bound on the capacity of the studied network. Doing the relevant cuts as displayed in Fig. \ref{fig:CutSet} and using bounding techniques (details in \cite{MyPaper}) will give

\begin{equation}
\begin{aligned}
R_{m_1}+R_{m_2}&\leq\frac{1}{2}\min\{\alpha_1,\beta_1,\delta_1,\psi_1\},\\
R_{m_3}+R_{m_4}&\leq\frac{1}{2}\min\{\alpha_2,\beta_2,\delta_2,\psi_2\},
\end{aligned}
\end{equation}
with
\begin{equation*}
\begin{aligned}
\alpha_1=&\log_2(1+\lambda_1 P_1)+\log_2(1+\lambda_2P_2),\\
\alpha_2=&\log_2(1+\lambda_1 P_{R_1})+\log_2(1+\lambda_2P_{R_2}),\\
\beta_1=&\log_2(1+|h_3|^2P_{R_1})+\log_2(1+|h_4|^2P_{R_2}),\\
\beta_2=&\log_2(1+|h_3|^2P_{U_1})+\log_2(1+|h_4|^2P_{U_2}),\\
\delta_1=&\log_2(1+||\boldsymbol{h}_1||^2P_1)+\log_2(1+|h_4|^2P_{R_2}),\\
\delta_2=&\log_2(1+||\boldsymbol{h}_1||^2P_{R_1})+\log_2(1+|h_4|^2P_{U_2}),\\
\psi_1=&\log_2(1+||\boldsymbol{h}_2||^2 P_2)+\log_2(1+|h_3|^2P_{R_2}),\\
\psi_2=&\log_2(1+||\boldsymbol{h}_2||^2 P_{R_2})+\log_2(1+|h_3|^2P_{U_1}),
\end{aligned}
\end{equation*}
where $\alpha_i,\beta_i,\delta_i,\psi_i$ have been derived from Cuts 1-4, respectively, and $\lambda_1,\lambda_2$ are the eigenvalues of $\boldsymbol{HH}^H$ with \begin{equation*}
\boldsymbol{H}=\begin{bmatrix}
\boldsymbol{h}_1 & \boldsymbol{h}_2
\end{bmatrix}^H.
\end{equation*} The factor of $\frac{1}{2}$ again originates from half-duplex mode. For details we refer the interested reader to \cite{MyPaper}.
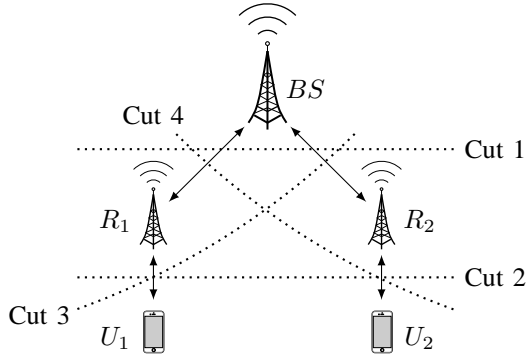
\begin{figure}[!tb]
\centering
\begin{tikzpicture}
\basestation{0}{3}{0.5}
\node[color=black] at (0.5,3) {$BS$};
\basestation{-1.5}{1.25}{0.35}
\node[color=black] at (-2,1.25) {$R_1$};
\basestation{1.5}{1.25}{0.35}
\node[color=black] at (2,1.25) {$R_2$};
\iPhone{-1.63}{-0.5}{0.4}
\node[color=black] at (-2,-0.3) {$U_1$};
\iPhone{1.375}{-0.5}{0.4}
\node[color=black] at (2,-0.3) {$U_2$};
\draw[<->,>=latex] (-0.3,2.5) --(-1.3,1.5);
\draw[<->,>=latex] (0.3,2.5) --(1.3,1.5);
\draw[<->,>=latex] (-1.5,0.8) --(-1.5,0.2);
\draw[<->,>=latex] (1.5,0.8) --(1.5,0.2);
\draw[thick,dotted] (-2.5,2.2) --(2.5,2.2);
\draw[thick,dotted] (-2.5,0.5) --(2.5,0.5);
\draw [thick,dotted](-1.2,2.4) arc (225:250:10cm);
\draw [thick,dotted](-2.5,0.08) arc (290:315:10cm);
\node[color=black] at (3,2.2) {Cut 1};
\node[color=black] at (3,0.5) {Cut 2};
\node[color=black] at (-3,0) {Cut 3};
\node[color=black] at (-1.5,2.65) {Cut 4};
\end{tikzpicture}
\caption{Discussed setup with the cuts made to establish an upper bounds on the capacity.}
\label{fig:CutSet}
\end{figure}

\section{Numerical Results and Simulation}\label{section7}

\subsection{Optimal Power Allocation}
Fig. \ref{fig:Simu1} shows the powers of the macrocell BS and the mobile users as variable instances among the x- and y-axis, while the powers of the femtocell BSs are fixed $(P_{R_1},P_{R_2}=5)$. It can be seen that the overall rate $R_\Sigma$, which is depicted along the z-axis, is constrained by the resources of the macrocell BS and/or the mobile users for small values of $P_{BS},~P_{U_1}$ and $P_{U_2}$. For larger values of $P_{BS},~P_{U_1}$ and $P_{U_2}$ we can conclude that the message rates of the setup are constrained by the resources of the femtocell BSs. 

\begin{figure}[htb]
\begin{tikzpicture}

\begin{axis}[%
width=0.45*0.4*13.907in,
height=0.5*0.2*8.793in,
at={(0.809in,0.513in)},
scale only axis,
every outer x axis line/.append style={black},
every x tick label/.append style={font=\color{black}},
xmin=0,
xmax=15,
tick align=outside,
xlabel={$P_{U_1},P_{U_2}$},
xmajorgrids,
every outer y axis line/.append style={black},
every y tick label/.append style={font=\color{black}},
ymin=0,
ymax=20,
ylabel={$P_{BS}$},
ymajorgrids,
every outer z axis line/.append style={black},
every z tick label/.append style={font=\color{black}},
zmin=0,
zmax=4,
zlabel={$R_\Sigma\text{ [bits]}$},
zmajorgrids,
view={-37.5}{30},
axis background/.style={fill=white},
axis x line*=bottom,
axis y line*=left,
axis z line*=right
]

\addplot3[%
surf,
shader=flat corner,draw=black,z buffer=sort,colormap/jet,mesh/rows=16]
table[row sep=crcr, point meta=\thisrow{c}] {%
x	y	z	c\\
0	0	0	0\\
0	1	0.586491598521847	0.586491598521847\\
0	2	1.28988545190754	1.28988545190754\\
0	3	1.46380672536551	1.46380672536551\\
0	4	1.46380672536551	1.46380672536551\\
0	5	1.46380672536551	1.46380672536551\\
0	6	1.46380672536551	1.46380672536551\\
0	7	1.46380672536551	1.46380672536551\\
0	8	1.46380672536551	1.46380672536551\\
0	9	1.46380672536551	1.46380672536551\\
0	10	1.46380672536551	1.46380672536551\\
0	11	1.46380672536551	1.46380672536551\\
0	12	1.46380672536551	1.46380672536551\\
0	13	1.46380672536551	1.46380672536551\\
0	14	1.46380672536551	1.46380672536551\\
0	15	1.46380672536551	1.46380672536551\\
0	16	1.46380672536551	1.46380672536551\\
0	17	1.46380672536551	1.46380672536551\\
0	18	1.46380672536551	1.46380672536551\\
0	19	1.46380672536551	1.46380672536551\\
0	20	1.46380672536551	1.46380672536551\\
1	0	0.424534757767006	0.424534757767006\\
1	1	1.01102635628885	1.01102635628885\\
1	2	1.71442020967455	1.71442020967455\\
1	3	1.88834148313252	1.88834148313252\\
1	4	1.88834148313252	1.88834148313252\\
1	5	1.88834148313252	1.88834148313252\\
1	6	1.88834148313252	1.88834148313252\\
1	7	1.88834148313252	1.88834148313252\\
1	8	1.88834148313252	1.88834148313252\\
1	9	1.88834148313252	1.88834148313252\\
1	10	1.88834148313252	1.88834148313252\\
1	11	1.88834148313252	1.88834148313252\\
1	12	1.88834148313252	1.88834148313252\\
1	13	1.88834148313252	1.88834148313252\\
1	14	1.88834148313252	1.88834148313252\\
1	15	1.88834148313252	1.88834148313252\\
1	16	1.88834148313252	1.88834148313252\\
1	17	1.88834148313252	1.88834148313252\\
1	18	1.88834148313252	1.88834148313252\\
1	19	1.88834148313252	1.88834148313252\\
1	20	1.88834148313252	1.88834148313252\\
2	0	0.905596299083864	0.905596299083864\\
2	1	1.49208789760571	1.49208789760571\\
2	2	2.19548175099141	2.19548175099141\\
2	3	2.36940302444938	2.36940302444938\\
2	4	2.36940302444938	2.36940302444938\\
2	5	2.36940302444938	2.36940302444938\\
2	6	2.36940302444938	2.36940302444938\\
2	7	2.36940302444938	2.36940302444938\\
2	8	2.36940302444938	2.36940302444938\\
2	9	2.36940302444938	2.36940302444938\\
2	10	2.36940302444938	2.36940302444938\\
2	11	2.36940302444938	2.36940302444938\\
2	12	2.36940302444938	2.36940302444938\\
2	13	2.36940302444938	2.36940302444938\\
2	14	2.36940302444938	2.36940302444938\\
2	15	2.36940302444938	2.36940302444938\\
2	16	2.36940302444938	2.36940302444938\\
2	17	2.36940302444938	2.36940302444938\\
2	18	2.36940302444938	2.36940302444938\\
2	19	2.36940302444938	2.36940302444938\\
2	20	2.36940302444938	2.36940302444938\\
3	0	1.33573992447221	1.33573992447221\\
3	1	1.92223152299406	1.92223152299406\\
3	2	2.62562537637976	2.62562537637976\\
3	3	2.79954664983773	2.79954664983773\\
3	4	2.79954664983773	2.79954664983773\\
3	5	2.79954664983773	2.79954664983773\\
3	6	2.79954664983773	2.79954664983773\\
3	7	2.79954664983773	2.79954664983773\\
3	8	2.79954664983773	2.79954664983773\\
3	9	2.79954664983773	2.79954664983773\\
3	10	2.79954664983773	2.79954664983773\\
3	11	2.79954664983773	2.79954664983773\\
3	12	2.79954664983773	2.79954664983773\\
3	13	2.79954664983773	2.79954664983773\\
3	14	2.79954664983773	2.79954664983773\\
3	15	2.79954664983773	2.79954664983773\\
3	16	2.79954664983773	2.79954664983773\\
3	17	2.79954664983773	2.79954664983773\\
3	18	2.79954664983773	2.79954664983773\\
3	19	2.79954664983773	2.79954664983773\\
3	20	2.79954664983773	2.79954664983773\\
4	0	1.66481659019786	1.66481659019786\\
4	1	2.25130818871971	2.25130818871971\\
4	2	2.9547020421054	2.9547020421054\\
4	3	3.12862331556337	3.12862331556337\\
4	4	3.12862331556337	3.12862331556337\\
4	5	3.12862331556337	3.12862331556337\\
4	6	3.12862331556337	3.12862331556337\\
4	7	3.12862331556337	3.12862331556337\\
4	8	3.12862331556337	3.12862331556337\\
4	9	3.12862331556337	3.12862331556337\\
4	10	3.12862331556337	3.12862331556337\\
4	11	3.12862331556337	3.12862331556337\\
4	12	3.12862331556337	3.12862331556337\\
4	13	3.12862331556337	3.12862331556337\\
4	14	3.12862331556337	3.12862331556337\\
4	15	3.12862331556337	3.12862331556337\\
4	16	3.12862331556337	3.12862331556337\\
4	17	3.12862331556337	3.12862331556337\\
4	18	3.12862331556337	3.12862331556337\\
4	19	3.12862331556337	3.12862331556337\\
4	20	3.12862331556337	3.12862331556337\\
5	0	1.93184554445999	1.93184554445999\\
5	1	2.51833714298184	2.51833714298184\\
5	2	3.22173099636753	3.22173099636753\\
5	3	3.3956522698255	3.3956522698255\\
5	4	3.3956522698255	3.3956522698255\\
5	5	3.3956522698255	3.3956522698255\\
5	6	3.3956522698255	3.3956522698255\\
5	7	3.3956522698255	3.3956522698255\\
5	8	3.3956522698255	3.3956522698255\\
5	9	3.3956522698255	3.3956522698255\\
5	10	3.3956522698255	3.3956522698255\\
5	11	3.3956522698255	3.3956522698255\\
5	12	3.3956522698255	3.3956522698255\\
5	13	3.3956522698255	3.3956522698255\\
5	14	3.3956522698255	3.3956522698255\\
5	15	3.3956522698255	3.3956522698255\\
5	16	3.3956522698255	3.3956522698255\\
5	17	3.3956522698255	3.3956522698255\\
5	18	3.3956522698255	3.3956522698255\\
5	19	3.3956522698255	3.3956522698255\\
5	20	3.3956522698255	3.3956522698255\\
6	0	2.13786631313436	2.13786631313436\\
6	1	2.72435791165621	2.72435791165621\\
6	2	3.4277517650419		3.4277517650419\\
6	3	3.60167303849987	3.60167303849987\\
6	4	3.60167303849987	3.60167303849987\\
6	5	3.60167303849987	3.60167303849987\\
6	6	3.60167303849987	3.60167303849987\\
6	7	3.60167303849987	3.60167303849987\\
6	8	3.60167303849987	3.60167303849987\\
6	9	3.60167303849987	3.60167303849987\\
6	10	3.60167303849987	3.60167303849987\\
6	11	3.60167303849987	3.60167303849987\\
6	12	3.60167303849987	3.60167303849987\\
6	13	3.60167303849987	3.60167303849987\\
6	14	3.60167303849987	3.60167303849987\\
6	15	3.60167303849987	3.60167303849987\\
6	16	3.60167303849987	3.60167303849987\\
6	17	3.60167303849987	3.60167303849987\\
6	18	3.60167303849987	3.60167303849987\\
6	19	3.60167303849987	3.60167303849987\\
6	20	3.60167303849987	3.60167303849987\\
7	0	2.24283390184567	2.24283390184567\\
7	1	2.82932550036751	2.82932550036751\\
7	2	3.53271935375321	3.53271935375321\\
7	3	3.70664062721118	3.70664062721118\\
7	4	3.70664062721118	3.70664062721118\\
7	5	3.70664062721118	3.70664062721118\\
7	6	3.70664062721118	3.70664062721118\\
7	7	3.70664062721118	3.70664062721118\\
7	8	3.70664062721118	3.70664062721118\\
7	9	3.70664062721118	3.70664062721118\\
7	10	3.70664062721118	3.70664062721118\\
7	11	3.70664062721118	3.70664062721118\\
7	12	3.70664062721118	3.70664062721118\\
7	13	3.70664062721118	3.70664062721118\\
7	14	3.70664062721118	3.70664062721118\\
7	15	3.70664062721118	3.70664062721118\\
7	16	3.70664062721118	3.70664062721118\\
7	17	3.70664062721118	3.70664062721118\\
7	18	3.70664062721118	3.70664062721118\\
7	19	3.70664062721118	3.70664062721118\\
7	20	3.70664062721118	3.70664062721118\\
8	0	2.30641808885602	2.30641808885602\\
8	1	2.89290968737787	2.89290968737787\\
8	2	3.59630354076356	3.59630354076356\\
8	3	3.77022481422153	3.77022481422153\\
8	4	3.77022481422153	3.77022481422153\\
8	5	3.77022481422153	3.77022481422153\\
8	6	3.77022481422153	3.77022481422153\\
8	7	3.77022481422153	3.77022481422153\\
8	8	3.77022481422153	3.77022481422153\\
8	9	3.77022481422153	3.77022481422153\\
8	10	3.77022481422153	3.77022481422153\\
8	11	3.77022481422153	3.77022481422153\\
8	12	3.77022481422153	3.77022481422153\\
8	13	3.77022481422153	3.77022481422153\\
8	14	3.77022481422153	3.77022481422153\\
8	15	3.77022481422153	3.77022481422153\\
8	16	3.77022481422153	3.77022481422153\\
8	17	3.77022481422153	3.77022481422153\\
8	18	3.77022481422153	3.77022481422153\\
8	19	3.77022481422153	3.77022481422153\\
8	20	3.77022481422153	3.77022481422153\\
9	0	2.30641808885602	2.30641808885602\\
9	1	2.89290968737787	2.89290968737787\\
9	2	3.59630354076356	3.59630354076356\\
9	3	3.77022481422153	3.77022481422153\\
9	4	3.77022481422153	3.77022481422153\\
9	5	3.77022481422153	3.77022481422153\\
9	6	3.77022481422153	3.77022481422153\\
9	7	3.77022481422153	3.77022481422153\\
9	8	3.77022481422153	3.77022481422153\\
9	9	3.77022481422153	3.77022481422153\\
9	10	3.77022481422153	3.77022481422153\\
9	11	3.77022481422153	3.77022481422153\\
9	12	3.77022481422153	3.77022481422153\\
9	13	3.77022481422153	3.77022481422153\\
9	14	3.77022481422153	3.77022481422153\\
9	15	3.77022481422153	3.77022481422153\\
9	16	3.77022481422153	3.77022481422153\\
9	17	3.77022481422153	3.77022481422153\\
9	18	3.77022481422153	3.77022481422153\\
9	19	3.77022481422153	3.77022481422153\\
9	20	3.77022481422153	3.77022481422153\\
10	0	2.30641808885602	2.30641808885602\\
10	1	2.89290968737787	2.89290968737787\\
10	2	3.59630354076356	3.59630354076356\\
10	3	3.77022481422153	3.77022481422153\\
10	4	3.77022481422153	3.77022481422153\\
10	5	3.77022481422153	3.77022481422153\\
10	6	3.77022481422153	3.77022481422153\\
10	7	3.77022481422153	3.77022481422153\\
10	8	3.77022481422153	3.77022481422153\\
10	9	3.77022481422153	3.77022481422153\\
10	10	3.77022481422153	3.77022481422153\\
10	11	3.77022481422153	3.77022481422153\\
10	12	3.77022481422153	3.77022481422153\\
10	13	3.77022481422153	3.77022481422153\\
10	14	3.77022481422153	3.77022481422153\\
10	15	3.77022481422153	3.77022481422153\\
10	16	3.77022481422153	3.77022481422153\\
10	17	3.77022481422153	3.77022481422153\\
10	18	3.77022481422153	3.77022481422153\\
10	19	3.77022481422153	3.77022481422153\\
10	20	3.77022481422153	3.77022481422153\\
11	0	2.30641808885602	2.30641808885602\\
11	1	2.89290968737787	2.89290968737787\\
11	2	3.59630354076356	3.59630354076356\\
11	3	3.77022481422153	3.77022481422153\\
11	4	3.77022481422153	3.77022481422153\\
11	5	3.77022481422153	3.77022481422153\\
11	6	3.77022481422153	3.77022481422153\\
11	7	3.77022481422153	3.77022481422153\\
11	8	3.77022481422153	3.77022481422153\\
11	9	3.77022481422153	3.77022481422153\\
11	10	3.77022481422153	3.77022481422153\\
11	11	3.77022481422153	3.77022481422153\\
11	12	3.77022481422153	3.77022481422153\\
11	13	3.77022481422153	3.77022481422153\\
11	14	3.77022481422153	3.77022481422153\\
11	15	3.77022481422153	3.77022481422153\\
11	16	3.77022481422153	3.77022481422153\\
11	17	3.77022481422153	3.77022481422153\\
11	18	3.77022481422153	3.77022481422153\\
11	19	3.77022481422153	3.77022481422153\\
11	20	3.77022481422153	3.77022481422153\\
12	0	2.30641808885602	2.30641808885602\\
12	1	2.89290968737787	2.89290968737787\\
12	2	3.59630354076356	3.59630354076356\\
12	3	3.77022481422153	3.77022481422153\\
12	4	3.77022481422153	3.77022481422153\\
12	5	3.77022481422153	3.77022481422153\\
12	6	3.77022481422153	3.77022481422153\\
12	7	3.77022481422153	3.77022481422153\\
12	8	3.77022481422153	3.77022481422153\\
12	9	3.77022481422153	3.77022481422153\\
12	10	3.77022481422153	3.77022481422153\\
12	11	3.77022481422153	3.77022481422153\\
12	12	3.77022481422153	3.77022481422153\\
12	13	3.77022481422153	3.77022481422153\\
12	14	3.77022481422153	3.77022481422153\\
12	15	3.77022481422153	3.77022481422153\\
12	16	3.77022481422153	3.77022481422153\\
12	17	3.77022481422153	3.77022481422153\\
12	18	3.77022481422153	3.77022481422153\\
12	19	3.77022481422153	3.77022481422153\\
12	20	3.77022481422153	3.77022481422153\\
13	0	2.30641808885602	2.30641808885602\\
13	1	2.89290968737787	2.89290968737787\\
13	2	3.59630354076356	3.59630354076356\\
13	3	3.77022481422153	3.77022481422153\\
13	4	3.77022481422153	3.77022481422153\\
13	5	3.77022481422153	3.77022481422153\\
13	6	3.77022481422153	3.77022481422153\\
13	7	3.77022481422153	3.77022481422153\\
13	8	3.77022481422153	3.77022481422153\\
13	9	3.77022481422153	3.77022481422153\\
13	10	3.77022481422153	3.77022481422153\\
13	11	3.77022481422153	3.77022481422153\\
13	12	3.77022481422153	3.77022481422153\\
13	13	3.77022481422153	3.77022481422153\\
13	14	3.77022481422153	3.77022481422153\\
13	15	3.77022481422153	3.77022481422153\\
13	16	3.77022481422153	3.77022481422153\\
13	17	3.77022481422153	3.77022481422153\\
13	18	3.77022481422153	3.77022481422153\\
13	19	3.77022481422153	3.77022481422153\\
13	20	3.77022481422153	3.77022481422153\\
14	0	2.30641808885602	2.30641808885602\\
14	1	2.89290968737787	2.89290968737787\\
14	2	3.59630354076356	3.59630354076356\\
14	3	3.77022481422153	3.77022481422153\\
14	4	3.77022481422153	3.77022481422153\\
14	5	3.77022481422153	3.77022481422153\\
14	6	3.77022481422153	3.77022481422153\\
14	7	3.77022481422153	3.77022481422153\\
14	8	3.77022481422153	3.77022481422153\\
14	9	3.77022481422153	3.77022481422153\\
14	10	3.77022481422153	3.77022481422153\\
14	11	3.77022481422153	3.77022481422153\\
14	12	3.77022481422153	3.77022481422153\\
14	13	3.77022481422153	3.77022481422153\\
14	14	3.77022481422153	3.77022481422153\\
14	15	3.77022481422153	3.77022481422153\\
14	16	3.77022481422153	3.77022481422153\\
14	17	3.77022481422153	3.77022481422153\\
14	18	3.77022481422153	3.77022481422153\\
14	19	3.77022481422153	3.77022481422153\\
14	20	3.77022481422153	3.77022481422153\\
15	0	2.30641808885602	2.30641808885602\\
15	1	2.89290968737787	2.89290968737787\\
15	2	3.59630354076356	3.59630354076356\\
15	3	3.77022481422153	3.77022481422153\\
15	4	3.77022481422153	3.77022481422153\\
15	5	3.77022481422153	3.77022481422153\\
15	6	3.77022481422153	3.77022481422153\\
15	7	3.77022481422153	3.77022481422153\\
15	8	3.77022481422153	3.77022481422153\\
15	9	3.77022481422153	3.77022481422153\\
15	10	3.77022481422153	3.77022481422153\\
15	11	3.77022481422153	3.77022481422153\\
15	12	3.77022481422153	3.77022481422153\\
15	13	3.77022481422153	3.77022481422153\\
15	14	3.77022481422153	3.77022481422153\\
15	15	3.77022481422153	3.77022481422153\\
15	16	3.77022481422153	3.77022481422153\\
15	17	3.77022481422153	3.77022481422153\\
15	18	3.77022481422153	3.77022481422153\\
15	19	3.77022481422153	3.77022481422153\\
15	20	3.77022481422153	3.77022481422153\\
};
\end{axis}
\end{tikzpicture}
\caption{Once the throughput of the system in constrained by the available power at the femtocell BSs, the use of more power at the macrocell BS or the mobile users does not result in a higher sum rate in the setup  (dark red plateau).}
\label{fig:Simu1}
\end{figure}
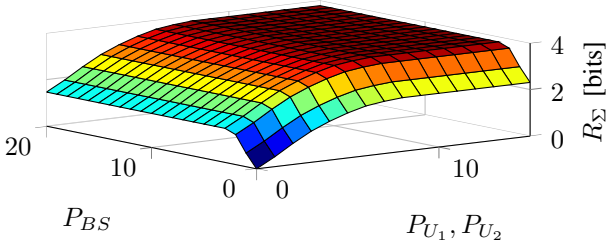

A further improvement of the uplink rate will not result in a higher overall rate. For example, one can see that choosing $P_{U_1},P_{U_2}=9$ and $P_{BS}=4$ leads to the same data throughput as letting $P_{U_1},P_{U_2},P_{BS}\rightarrow\infty$. Fig. \ref{fig:Simu1} confirms the statement made earlier, in which we claim that there are multiple power allocations that result in the maximum throughput of the system, but only one such allocation that consumes the least power.

\subsection{Performance Analysis}

To compare the ZF-PA not only with the cut-set bound, but also with existing schemes, we shall use the scheme of time-division multiple-access (TDMA) as a reference. The rate expressions for TDMA can be derived in a similar way like the ones for ZF-EPA, leading to

\begin{equation}
\begin{aligned}
R_{m_1}^{\text{TDMA}}&\leq\frac{1}{4}\min\Bigl\{\log(1+\gamma_{1,UL}),\log(1+\gamma_{1,DL})\Bigr\},\\
R_{m_2}^{\text{TDMA}}&\leq\frac{1}{4}\min\Bigl\{\log(1+\gamma_{2,UL}),\log(1+\gamma_{2,DL})\Bigr\},\\
R_{m_3}^{\text{TDMA}}&\leq\frac{1}{4}\min\Bigl\{\log(1+|h_3|^2P_{U_1}),\log(1+\gamma_{3,DL})\Bigr\},\\
R_{m_4}^{\text{TDMA}}&\leq\frac{1}{4}\min\Bigl\{\log(1+|h_4|^2P_{U_2}),\log(1+\gamma_{4,DL})\Bigr\},
\end{aligned}
\end{equation}
where the factor of $\frac{1}{4}$ originates from the four time slots needed to transmit all messages.\\
Fig. \ref{fig:Gap} (a) shows a typical performance of ZF-EPA in comparison with TDMA and the cut-set bound derived in section \ref{Section6}. The simulation was done with $M=5$ antennas at the macrocell BS and channels chosen randomly and independently according to $\boldsymbol{h}_1,\boldsymbol{h}_2\sim\mathcal{N}(\boldsymbol{0},\boldsymbol{I}),~h_3,h_4\sim\mathcal{N}(0,1)$. The parameter $P$, which is depicted among the y-axis, is used to match the relative powers available at each device of the setup. These have been chosen such that

\begin{equation}
\begin{aligned}
P_{BS,\text{MAX}}&=P,\\
P_{R_1,\text{MAX}},P_{R_2,\text{MAX}}&=\nicefrac{1}{2}~P,\\
P_{U_1,\text{MAX}},P_{U_2,\text{MAX}}&=\nicefrac{1}{4}~P.
\end{aligned}
\end{equation}  

It becomes evident that ZF-EPA outperforms TDMA significantly in the mid- to high-power regime by a factor of 2. This is because it requires only half the amount of time slots to complete the transmission process. Fig. \ref{fig:Gap} (b) gives a more detailed insight on the gap between the cut-set bound and ZF-EPA. Although spiking at $0.6$ bit in the low power regime, the gap becomes no larger than $0.25$ bit in the mid-power regime. It is around $0.2$ bit in the high-power regime, which is under $6~\%$ of the capacity.

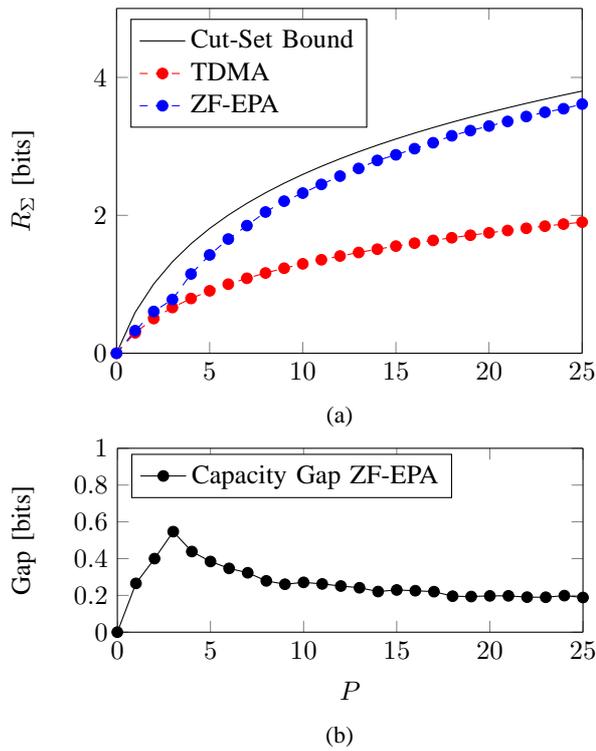
\begin{figure}[htb]
\begin{subfigure}[htb]{1\linewidth}
\begin{tikzpicture}

\begin{axis}[%
width=2.411in,
height=0.9*2in,
at={(0.809in,0.513in)},
scale only axis,
separate axis lines,
every outer x axis line/.append style={black},
every x tick label/.append style={font=\color{black}},
xmin=0,
xmax=25,
every outer y axis line/.append style={black},
every y tick label/.append style={font=\color{black}},
ymin=0,
ymax=5,
ylabel={$R_\Sigma\text{ [bits]}$},
axis background/.style={fill=white},
legend style={at={(0.03,0.97)},anchor=north west,legend cell align=left,align=left,draw=black}
]
\addplot [color=black,mark options={solid}]
  table[row sep=crcr]{%
0	0\\
1	0.59248847516809\\
2	1.0059918487876\\
3	1.32607655016805\\
4	1.588035446936\\
5	1.81013436113522\\
6	2.00312605532987\\
7	2.17389914913995\\
8	2.32714128433977\\
9	2.46619072587138\\
10	2.59351092354512\\
11	2.71097264823393\\
12	2.82003055632132\\
13	2.92183839830429\\
14	3.0173268485546\\
15	3.10725765015285\\
16	3.19226224280289\\
17	3.2728699291622\\
18	3.34952880971306\\
19	3.4226216082055\\
20	3.49247781624479\\
21	3.55938313984824\\
22	3.62358693737335\\
23	3.68530814090818\\
24	3.74474001796225\\
25	3.80205403595704\\
};
\addlegendentry{Cut-Set Bound};

\addplot [color=red,dashed,mark=*,mark options={solid}]
  table[row sep=crcr]{%
0	0\\
1	0.296244237584045\\
2	0.502995924393798\\
3	0.663038275084027\\
4	0.794017723468002\\
5	0.905067180567608\\
6	1.00156302766493\\
7	1.08694957456997\\
8	1.16357064216989\\
9	1.23309536293569\\
10	1.29675546177256\\
11	1.35548632411696\\
12	1.41001527816066\\
13	1.46091919915214\\
14	1.5086634242773\\
15	1.55362882507643\\
16	1.59613112140144\\
17	1.6364349645811\\
18	1.67476440485653\\
19	1.71131080410275\\
20	1.74623890812239\\
21	1.77969156992412\\
22	1.81179346868667\\
23	1.84265407045409\\
24	1.87237000898113\\
25	1.90102701797852\\
};
\addlegendentry{TDMA};

\addplot [color=blue,dashed,mark=*,mark options={solid}]
  table[row sep=crcr]{%
0	0\\
1	0.327191348122052\\
2	0.605908199425258\\
3	0.779243907460405\\
4	1.14948214184662\\
5	1.42561459689886\\
6	1.65616168946185\\
7	1.85082025118313\\
8	2.04814053615669\\
9	2.20547400770713\\
10	2.32258775747421\\
11	2.4486124124349\\
12	2.56901622311045\\
13	2.68027183630003\\
14	2.79628428788117\\
15	2.87773650661882\\
16	2.96676827421166\\
17	3.05256589148766\\
18	3.15380090769103\\
19	3.22881223316789\\
20	3.29503424437264\\
21	3.36182981827173\\
22	3.43308537638025\\
23	3.49527106210598\\
24	3.54587501304256\\
25	3.6135215681366\\
};
\addlegendentry{ZF-EPA};

\end{axis}
\end{tikzpicture}%
\caption{ }
\end{subfigure}

\begin{subfigure}[htb]{1\linewidth}
\begin{tikzpicture}

\begin{axis}[%
width=2.411in,
height=0.8*1.2in,
at={(0.809in,0.513in)},
scale only axis,
separate axis lines,
every outer x axis line/.append style={black},
every x tick label/.append style={font=\color{black}},
xmin=0,
xmax=25,
xlabel={$P$},
every outer y axis line/.append style={black},
every y tick label/.append style={font=\color{black}},
ymin=0,
ymax=1,
ylabel={$\text{Gap [bits]}$},
axis background/.style={fill=white},
legend style={at={(0.03,0.97)},anchor=north west,legend cell align=left,align=left,draw=black}
]

\addplot [color=black,solid,mark=*,mark options={solid}]
  table[row sep=crcr]{%
0	0\\
1	0.2653\\
2	0.4001\\
3	0.5468\\
4	0.4386\\
5	0.3845\\
6	0.3470\\
7	0.3231\\
8	0.2790\\
9	0.2607\\
10	0.2709\\
11	0.2624\\
12	0.2510\\
13	0.2416\\
14	0.2210\\
15	0.2295\\
16	0.2255\\
17	0.2203\\
18	0.1957\\
19	0.1938\\
20	0.1974\\
21	0.1976\\
22	0.1905\\
23	0.1900\\
24	0.1989\\
25	0.1885\\
};
\addlegendentry{Capacity Gap ZF-EPA};
\end{axis}
\end{tikzpicture}%
\caption{ }
\end{subfigure}
\caption{Performance analysis of the ZF-EPA and TDMA schemes with respect to the cut-set bound.}
\label{fig:Gap}
\end{figure}



%


\ifCLASSOPTIONcaptionsoff
  \newpage
\fi



%






\end{document}